\documentclass[letterpaper, 10 pt, conference]{ieeeconf}  

\IEEEoverridecommandlockouts                              

\usepackage{cite}
\usepackage{amsmath}
\usepackage{amssymb}
\usepackage{stackengine}
\usepackage{mathrsfs}
\usepackage{bbm}
\usepackage{mathtools}
\usepackage{bm}
\usepackage{amsfonts}
\usepackage{graphicx}
\usepackage{float}
\usepackage{url}
\usepackage{multirow}
\usepackage{xcolor}
\usepackage{algorithm}
\usepackage{algpseudocode}
\usepackage{graphicx}
\usepackage{subcaption}
\usepackage{epstopdf}
\usepackage{lipsum}
\usepackage{booktabs}
\usepackage{threeparttable}
\usepackage{url}
\usepackage{hyperref}
\hypersetup{hidelinks}

\hypersetup{
	colorlinks=true,
	linkcolor=black,
	citecolor=black,
	urlcolor=blue,
}

\graphicspath{{./}{Fig/}}

\DeclareMathOperator*{\argmin}{argmin}
\newtheorem{lemma}{Lemma}

\newtheorem{proposition}{Proposition}
\newtheorem{theorem}{Theorem}

\newtheorem{assumption}{Assumption}

	\title{\LARGE \bf Learning-based Homothetic Tube MPC}
	\author{Yulong Gao, Shuhao Yan, Jian Zhou, and Mark Cannon
        \thanks{This work was funded in part by Deutsche Forschungsgemeinschaft (DFG, German Research Foundation) under Germany's Excellence Strategy - EXC 2075 – 390740016; and the Strategic Research Area at Link\"oping-Lund in Information Technology (ELLIIT). S. Yan acknowledges the support by the Stuttgart Center for Simulation Science (SimTech).}%
		\thanks{Y. Gao is with the Department of Electrical and Electronic Engineering, Imperial College London, UK {\tt\small yulong.gao@imperial.ac.uk}}%
		\thanks{S. Yan is with the Department of Mathematics, University of Stuttgart, Germany {\tt\small shuhao.yan@mathematik.uni-stuttgart.de}}%
		\thanks{J. Zhou is with the  Department of Electrical Engineering, Link\"oping University, Sweden {\tt\small jian.zhou@liu.se}}%
   \thanks{M. Cannon is with the  Department of Engineering Science, University of Oxford, UK {\tt\small mark.cannon@eng.ox.ac.uk}}%
   }
   
\begin{document}	
\maketitle
\thispagestyle{empty} 
\pagestyle{empty}

 \begin{abstract}
In this paper, we study homothetic tube  model predictive control (MPC) of discrete-time linear systems subject to bounded additive disturbance and mixed constraints on the state and input.
 Different from most existing work on robust MPC, we assume that the true disturbance set is unknown but a conservative surrogate is available a priori.  Leveraging the real-time data, we develop 
 an online learning algorithm to approximate the true disturbance set. This approximation and the corresponding constraints in the MPC optimisation are updated online using computationally convenient linear programs.
 We provide statistical gaps between the true and learned disturbance sets, based on which, probabilistic recursive feasibility of homothetic tube MPC problems is discussed. Numerical simulations are provided to demonstrate the efficacy of our proposed algorithm and compare with state-of-the-art MPC algorithms.
 \end{abstract}

	
\section{Introduction}

Robust model predictive control (MPC) ensures constraint satisfaction under all possible realisations of uncertainty. It has received considerable interest, especially tube-based methods \cite{langson2004robust,mayne2005robust,rakovic2012homothetic}, due to their reduced online computational requirements and constraint handling. 
However, they suffer from potential conservativeness. Tubes bounding uncertainty evolution may not be exact, and the disturbance bounds used in constructing tubes can be overly conservative, leading to small feasible sets. In addition, the offline tube construction renders it difficult to exploit online information to further quantify uncertainty and update the tube over time.

\begin{figure}
	\centering
	{\includegraphics[width =0.6\columnwidth]{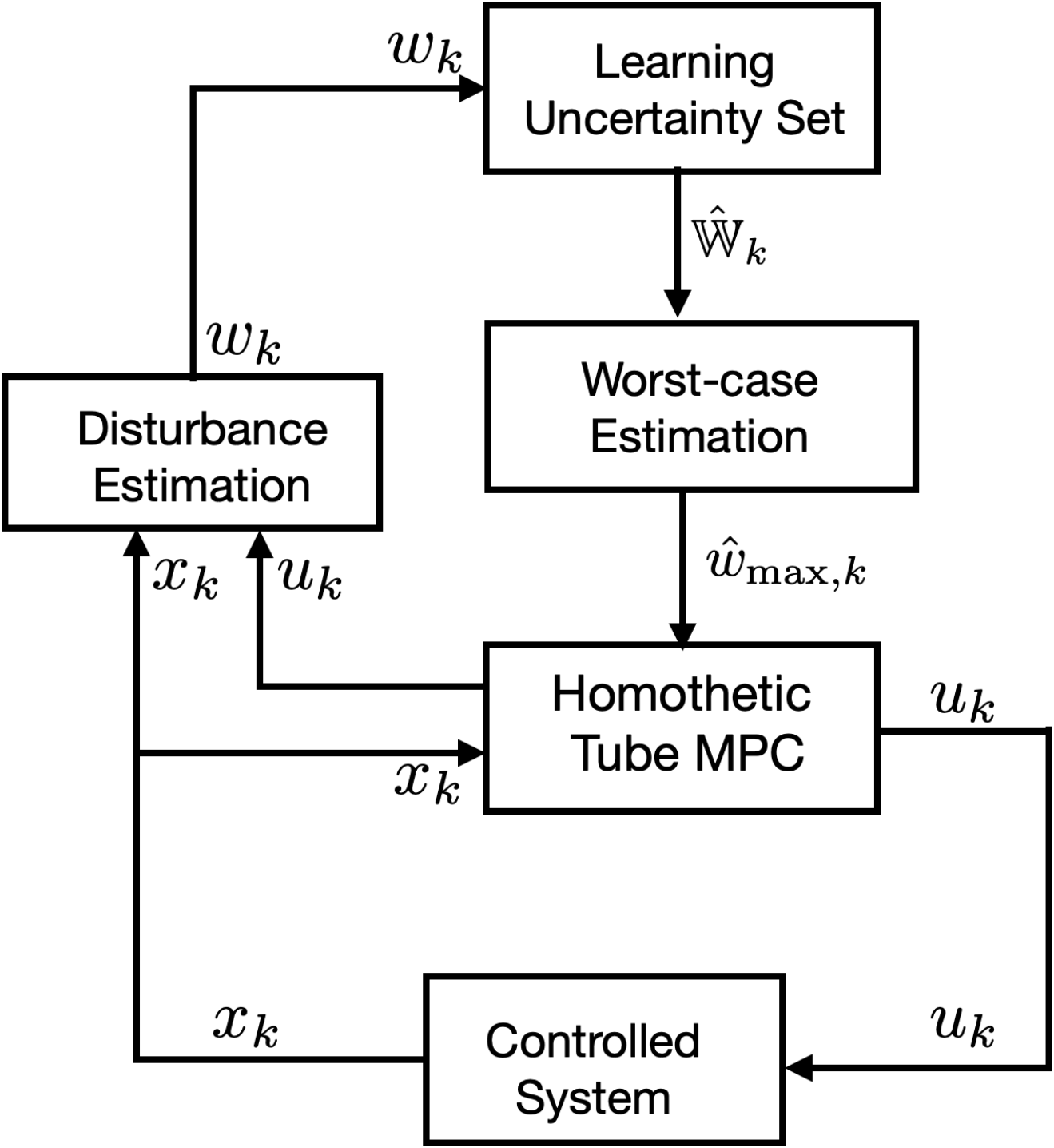}}
	\caption{Learning-based tube MPC framework, where $x_k$, $u_k$, $w_k$ are state, control input, and disturbance respectively; $\hat{\mathbb{W}}_k$ is the learned disturbance set; and $\hat{w}_{\max,k}$ is the worst-case estimation. }
	\label{Fig:learningtubeMPC}
\end{figure}

To mitigate the conservativeness discussed above, we propose a learning-based MPC algorithm incorporating online uncertainty quantification for discrete-time linear systems subject to bounded additive disturbance, as shown in Fig.~\ref{Fig:learningtubeMPC}. Since the exact set, $\mathbb{W}_{\rm true}$, containing all possible realisations of the disturbance is often unknown, we instead assume an over-approximate polytopic set, $\mathbb{W}$, of $\mathbb{W}_{\rm true}$ is given. Assuming perfect state feedback, we can collect disturbance measurements and use this data to design parameterised sets $\hat{\mathbb{W}}_k$ for approximating $\mathbb{W}_{\rm true}$ online via the scenario approach \cite{campi2008exact}. These parameterised sets are referred to as learned disturbance sets. 
The main contributions of this paper are summarised as follows. 1) The set $\hat{\mathbb{W}}_k$ is parameterised such that it scales $\mathbb{W}$ in different directions with heterogeneous parameters. This can yield a more accurate estimate of $\mathbb{W}_{\rm true}$ than the set parameterisation in \cite{gao2023learning}, which only allows for uniform scaling. Also, the parameterisation of $\hat{\mathbb{W}}_k$ enables the approximation problem to be efficiently solved as a linear program (LP).
2) The learned disturbance sets are updated online and incorporated in the synthesis of a computationally efficient MPC algorithm. In the MPC problem, constraints on predicted trajectories over an infinite horizon are enforced by a finite number of constraints. This number is fixed by leveraging the nestedness property of the time-varying learned disturbance sets. By contrast, in the previous work \cite{gao2023learning}, the required number of constraints needs to be recomputed at each time step in the learning-based rigid tube MPC problem, increasing computational burdens.

\textit{Related work}: Tube-based MPC provides a computationally tractable framework for control design, in which a main focus is tube parameterisations. Typical parameterisations include the rigid tube \cite{mayne2005robust} and homothetic tube \cite{rakovic2012homothetic}. Further developments revolve around the trade-off between computational complexity and control performance. In \cite{rakovic2012parameterized}, a tube parameterisation is proposed with cross-sections consisting of partial cross-sections expressed as convex hulls of finite numbers of points. The elastic tube \cite{7525471} generalises the homothetic tube, and its cross-section is determined by scaling a base shape set heterogeneously in all directions. 

Development of data-driven techniques has allowed improved methods for dealing with uncertainty in MPC design, for example, in terms of constraint handling and model description. In \cite{9206383}, the scenario approach is used to construct an inner convex approximation of a chance-constrained set, ensuring computational tractability of the resulting MPC algorithm. Similarly, support vector clustering is used to identify a high-density region of disturbance in stochastic MPC problems \cite{8727721,shang2019data} and results in robust optimisation problems. In \cite{8619572}, Gaussian process (GP) regression is applied to learn the unknown part of system dynamics with trajectory data, and a statistical GP model is used in the subsequent MPC design. In \cite{oestreich2023tube}, uncertainty quantification is combined with rigid tube MPC synthesis. The shape matrix of the ellipsoidal tube cross-section is chosen to be the covariance matrix of disturbance that is estimated from disturbance measurements and updated over time, aiming to compensate for the possibly imperfect knowledge of disturbance sets.

\textit{Notation}: We use $\mathbb{N}$ to denote the set of nonnegative integers and $\mathbb{R}_{+}$ to denote the set of nonnegative real numbers.  For some $q,s \in \mathbb{N}$ and $q<s$, let $\mathbb{N}_{[q,s]}:=\{r \in \mathbb{N}\mid q \leq r\leq s\}$. For two sets $\mathbb{X}$ and $\mathbb{Y}$, $\mathbb{X}\oplus \mathbb{Y}:=\{x+y\mid x\in\mathbb{X}, y\in \mathbb{Y}\}$. When $\leq$, $\geq$, $<$, and $>$ are applied to vectors, they are interpreted element-wise. Matrices of appropriate dimension with all elements equal to 1 and 0 are denoted by $\bm{1}$ and $\bm{0}$, respectively. 
The cardinality of a set $\mathbb{X}$ is denoted by $|\mathbb{X}|$ and ${\rm vol}(\mathbb{X})$ denotes its volume.

\section{Problem Description} \label{sec:problem description}

Consider a discrete-time linear system with additive disturbance in the form of
	\begin{eqnarray} \label{LTIsystem}
		x_{k+1} = Ax_{k}+Bu_{k}+w_k, 
	\end{eqnarray}
	where $x_k \in \mathbb{R}^{n_{x}}$ is the state,  $u_k\in \mathbb{R}^{n_{u}}$ the control input, and $w_k\in\mathbb{W}_{\rm true}(\subseteq\mathbb{W})$ the disturbance. The pair $(A,B)$ is assumed given and stabilisable, but the \emph{true disturbance set} $\mathbb{W}_{\rm true}$ is unknown. The pair $(x_k,u_k)$ is subject to
\begin{eqnarray}
		Fx_k+Gu_k\leq \bm{1}, \nonumber
\end{eqnarray}
	where $F\in \mathbb{R}^{n_c \times n_x}$ and $G\in  \mathbb{R}^{n_c \times n_u}$.

\begin{assumption}\label{Ass}
The set $\mathbb{W}$ is a known compact polytope given by $\{w\in \mathbb{R}^{n_x}\mid  V_w w\leq \bm{1}\}$, where $V_w\in \mathbb{R}^{n_v\times n_x}$.
\end{assumption}
	
\subsection{Homothetic Tube MPC based on Disturbance Set $\mathbb{W}$} \label{sec: preliminary}
	This section presents preliminaries of homothetic tube MPC. We begin with decomposing predicted trajectories as
	\begin{subequations}
		\begin{align} 
			x_{i|k} &= s_{i|k}+e_{i|k},  \
			u_{i|k} = Kx_{i|k}+c_{i|k},  \nonumber \\
			s_{i+1|k} &= \Phi s_{i|k}+Bc_{i|k}, \
			e_{i+1|k} = \Phi e_{i|k}+w_{i|k}, \nonumber
		\end{align} 
	\end{subequations}
	where $w_{i|k}\in \mathbb{W}$.
	Here $s_{i|k}$ and $e_{i|k}$ typically denote nominal and uncertain components of the predicted state, respectively, 
	and $K$ is fixed such that $\Phi:=A+BK$ is strictly stable. The free variable is $\bm{c}_k:=(c_{0|k}, \ldots,  c_{N-1|k})$, and we set $c_{i|k}=\bm{0}$ for $i\in \mathbb{N}_{\geq N}$, where $N$ is a prediction horizon.

Based on the decomposition, one can construct a homothetic tube $\{\mathbb{S}_{i|k}:=\alpha_{i|k}\mathbb{S}, i\in \mathbb{N}\}$ with \emph{nonnegative} scaling  parameters $\alpha_{i|k}$ to bound uncertainty evolution. The base set $\mathbb{S}$ is a convex and compact polytope given by 
$
	\mathbb{S}=\{e\in \mathbb{R}^{n_x}\mid  V_{s} e\leq \bm{1}\},
$
where $V_{s} \in \mathbb{R}^{n_s \times n_x}$.
\begin{assumption}\label{ass:robust invariance}
    The set $\mathbb{S}$ satisfies that $\Phi \mathbb{S} \oplus \mathbb{W} \subseteq \mathbb{S}$.
\end{assumption}

To ensure $e_{i|k}\in\mathbb{S}_{i|k}$ for all $i\in\mathbb{N}$, we impose that
\[
e_{0|k}\in\mathbb{S}_{0|k},\quad \Phi \alpha_{i|k}\mathbb{S} \oplus \mathbb{W}\subseteq \alpha_{i+1|k} \mathbb{S},~~\forall i \in \mathbb{N},
\]
where the second set of conditions is equivalent to
\begin{eqnarray} \label{Eq:alpha inequality}
	\alpha_{i|k} e_{\rm max} +w_{\rm max}\leq   \alpha_{i+1|k}\bm{1}, \ \forall i\in\mathbb{N},
\end{eqnarray}
with
$e_{\rm max}:=\max_{e\in \mathbb{S}}  V_{s} \Phi e$ and $w_{\rm max}:=\max_{w\in \mathbb{W}}  V_s w$. To limit the number of variables, one can set $\alpha_{i|k}=1$ for all $i\in\mathbb{N}_{\geq N}$, which satisfies condition \eqref{Eq:alpha inequality} under Assumption \ref{ass:robust invariance}.

Then, with this tube, constraints $Fx_{i|k}+Gu_{i|k}\leq \bm{1}$ for $i\in\mathbb{N}$ can be reformulated as 
\begin{equation}\label{eq:infinite deterministic constraints}
    \bar{F}\Psi^i \begin{bmatrix}
				s_{0|k} \\
				\bm{c}_k
			\end{bmatrix} \leq \bm{1}-\alpha_{i|k}h, \ i=0,1,\ldots, 
\end{equation}
where $h:=\max\limits_{e\in \mathbb{S}} (F+GK)e$, 
$\bar{F}:=\begin{bmatrix}
F+GK & GE
\end{bmatrix}$, $\Psi:=\begin{bmatrix}
\Phi & BE \\ \
\bm{0} & M
\end{bmatrix}$,  $E:=\begin{bmatrix}
I_{n_u} & \bm{0} & \cdots & \bm{0}
\end{bmatrix}$, and $M$ is the block-upshift operator. However, condition \eqref{eq:infinite deterministic constraints} involves infinitely many constraints. Lemma \ref{lemma:condition on tau} addresses this issue by providing an equivalent set defined by a finite number of constraints. 
\begin{lemma}[\!\!{\cite[\S3.7, ex.~9]{kouvaritakis2016model}}] \label{lemma:condition on tau}
    If $n \geq N-1$ is an integer such that $\bar{F}\Psi^{n+1} z \leq \bm{1} -h$ holds for all $(\alpha_0,\ldots,\alpha_{N-1})\in\mathbb{R}_+^N$ and $z$ satisfying%
    \begin{subequations} \label{eq:finite number of constraints}
        \begin{align} 
        &\alpha_i e_{\rm max} + w_{\rm max} \leq \alpha_{i+1} \bm{1}, \ i\in\mathbb{N}_{[0,N-1]},  \\
        &\bar{F}\Psi^i z \leq \bm{1}- \alpha_i h, i\in\mathbb{N}_{[0,n]}, ~~ \alpha_i=1,\ i\in\mathbb{N}_{\geq N},
        \end{align}
    \end{subequations}
    then the feasible set for $z$ and $(\alpha_0,\ldots,\alpha_{N-1})$ in (\ref{eq:finite number of constraints}a,b) is equal to the feasible set for \eqref{Eq:alpha inequality} and \eqref{eq:infinite deterministic constraints} with $[s_{0|k}^T \ \bm{c}_k^T] = z^T$ and $\alpha_{i|k}=\alpha_i$ for all $i\in \mathbb{N}$.
\end{lemma}

The cost function can be chosen as $$J(s_{0|k}, \bm{c}_k,\bm{\alpha}_k):=	\|s_{0|k}\|^2_{P_x}+\|\bm{c}_k\|^2_{P_c} + q_{\alpha}\sum_{i=0}^{N-1}(\alpha_{i|k}-1)^2,$$ where $\bm{\alpha}_k:=(\alpha_{0|k},\ldots,\alpha_{N-1|k})$, and positive semidefinite weighting matrices $P_x$, $P_c$ and positive weight parameter $q_{\alpha}$ are appropriately selected. 
To summarise, the MPC optimisation to be solved at time step $k$ is given by
\begin{eqnarray}\label{Eq:HomTubeMPCOPT}
 \begin{aligned}
     &\min\limits_{s_{0|k},\bm{c}_k,\bm{\alpha}_k\in\mathbb{R}^N_{+}} J(s_{0|k}, \bm{c}_k,\bm{\alpha}_k)    \\
		&{~\rm s.t.}
		\begin{cases}
			x_k-s_{0|k}\in \alpha_{0|k}\mathbb{S}, \\
			\alpha_{i|k} e_{\rm max} +w_{\rm max} \leq   \alpha_{i+1|k}\bm{1}, \forall i\in\mathbb{N}_{[0,N-1]},  \\
			\alpha_{i|k}=1,  \forall i\in\mathbb{N}_{\geq N}, \\
                \bar{F}\Psi^i \!\begin{bmatrix}
				s_{0|k} \\
				\bm{c}_k
			\end{bmatrix}\!\!\leq\! \bm{1}-\alpha_{i|k}h, 	 \forall i\in\mathbb{N}_{[0,\tau]}.
		\end{cases}
 \end{aligned}
\end{eqnarray}
We use ${\rm OPT}(\mathbb{W}, w_{\rm max},\tau)$ to denote these problems, where $\tau$ is chosen as the smallest integer such that the conditions of Lemma \ref{lemma:condition on tau} are satisfied with $n=\tau$. Problem \eqref{Eq:HomTubeMPCOPT} includes the optimisation for rigid tube MPC as a special case where cross-sections of the tube are fixed and $\alpha_{i|k}=1$ $\forall i\in\mathbb{N}$.
\begin{lemma}[\!\!{\cite[\S3.7, ex.~9]{kouvaritakis2016model}}]\label{lemma:conventional recursive feasibility}
    If the integer $\tau$ in \eqref{Eq:HomTubeMPCOPT} is chosen so that $n=\tau$ satisfies the conditions in Lemma~\ref{lemma:condition on tau}, then problem \eqref{Eq:HomTubeMPCOPT} is recursively feasible given initial feasibility.
\end{lemma}

\subsection{Motivation and Objective}
Due to reliance on a conservative disturbance set $\mathbb{W}$, the homothetic tube MPC \eqref{Eq:HomTubeMPCOPT} may have a small feasible region or even suffer from infeasibility. To address this, we propose to integrate learning of the true disturbance set $\mathbb{W}_{\rm true}$ with the design of homothetic tube MPC, as illustrated in Fig.~\ref{Fig:learningtubeMPC}. Learned disturbance sets are denoted by $\hat{\mathbb{W}}_k$.  

Our aim is two-fold: (1)  characterisation of the gap between $\mathbb{W}_{\rm true}$ and $\hat{\mathbb{W}}_k$; (2) theoretical analysis of the homothetic tube MPC algorithm based on $\hat{\mathbb{W}}_k$.

\section{Learning Uncertainty Set}\label{sec: learning uncertainty set}
In this section, we introduce a parameterised set for approximating $\mathbb{W}_{\rm true}$ using observed disturbance realisations and $\mathbb{W}$ as a base set. Based on scenario optimisation, we provide a probabilistic guarantee of approximation accuracy. 
\begin{assumption}\label{assumption: iid samples}
    The disturbance inputs $\{w_0,w_1,\ldots\}$ affecting the system~\eqref{LTIsystem} are i.i.d. according to an unknown probability distribution ${\rm Pr}$ with support $\mathbb{W}_{\rm true}$. 
\end{assumption}

We first define an information set $\mathcal{I}_k$ as a collection of observed disturbance realisations up to the $k$-th time step, so that $\mathcal{I}_k=\mathcal{I}_{k-1}\cup\{w_{k-1}\}$, where $w_{k-1}$ can be calculated given perfect state measurements. The information set $\mathcal{I}_0$ contains disturbance samples that are assumed available offline. The parameterised set is given by
\begin{equation}\label{Eq:W_parameterised homothetic}
	\mathcal{W}(v,\bm{\theta},\rho ):=	f(\mathbb{W},\bm{\theta}) \oplus (1-\rho)v,
\end{equation} 
where  $v\in \mathbb{W}$, $\bm{\theta}\in [0,1]^{n_v}$, and $\rho\in[0,1]$ are design parameters, and
$ f(\mathbb{W},\bm{\theta}):=\{w\in \mathbb{R}^{n_x}\mid  V_w w\leq \bm{\theta}\}$. The function $f(\mathbb{W},\bm{\theta})$ defines an operation to scale $\mathbb{W}$ in different directions with heterogeneous parameters. If $\bm{\theta}=\rho\bm{1}$, this becomes an uniform scaling and \eqref{Eq:W_parameterised homothetic} reduces to the parameterisation proposed in \cite{gao2023learning} for approximating the disturbance set. The set $\mathcal{W}$ enjoys the following properties.
\begin{lemma}\label{lemma:properties of parameterisation}
	For any $v\in \mathbb{W}$, $\bm{\theta}\in [0,1]^{n_v}$, and $\rho\in[0,1]$, if $\bm{\theta} \leq \rho\bm{1}$, then it holds that: (1) $\mathcal{W}(v,\bm{ \theta},\rho)\subseteq \mathbb{W}$, and (2) $( \theta_{\min})^{n_x} {\rm vol}(\mathbb{W}) \leq {\rm vol}(	\mathcal{W}(v,\bm{ \theta},\rho))\leq  ( \theta_{\max})^{n_x} {\rm vol}(\mathbb{W})$, where $ \theta_{\min}:=\min_{i}[\bm{ \theta}]_{i}$ and $ \theta_{\max}:=\max_{i}[\bm{ \theta}]_{i}$. 
\end{lemma}
\begin{proof}
(1) By its definition, $
\mathcal{W}(v,\bm{\theta},\rho)=\{w\in\mathbb{R}^{n_x}\, |\, V_w w \leq \bm{\theta} + (1-\rho)V_w v\}
$. Since $v \in \mathbb{W}$ and $1-\rho\geq0$, it holds that $(1-\rho)V_w v \leq (1-\rho)\bm{1}$. Under the condition $\bm{\theta}\leq \rho\bm{1}$, it follows that $\bm{\theta} + (1-\rho)V_w v \leq \bm{1}$. Thus, any $w\in \mathcal{W}(v,\bm{\theta},\rho)$ also lies in $\mathbb{W}$, which proves the set containment; (2) As the volume of a set remains unchanged after translation, it holds that ${\rm vol}(\mathcal{W}(v,\bm{\theta},\rho))={\rm vol}(f(\mathbb{W},\bm{\theta}))$. Also, it is clear that $f(\mathbb{W},\theta_{\rm min} \bm{1}) \subseteq f(\mathbb{W},\bm{\theta})\subseteq f(\mathbb{W},\theta_{\rm max} \bm{1})$. This implies that $(\theta_{\rm min})^{n_x}{\rm vol}(\mathbb{W}) = {\rm vol}(f(\mathbb{W},\theta_{\rm min} \bm{1})) \leq {\rm vol}(f(\mathbb{W},\bm{\theta}))  \leq {\rm vol}(f(\mathbb{W},\theta_{\rm max} \bm{1})) =(\theta_{\rm max})^{n_x} {\rm vol}(\mathbb{W})$, and the desired inequalities follow.
\end{proof}

\subsection{Initialisation}
We next compute minimal sets in the form of \eqref{Eq:W_parameterised homothetic} while containing all points in the information set, and such sets are not more conservative than $\mathbb{W}$ as ensured by Lemma \ref{lemma:properties of parameterisation}. This is formulated as the following optimisation problem at the initial time step
\begin{align} 
    &\min_{v,\bm{\theta},\rho} \bm{1}^T  \bm{\theta} + \rho  \label{quanset for homothetic mpc}\\
    &\begin{aligned} \text{ s.t. }~
        &V_w w^s \leq \bm{\theta} + (1-\rho)V_w v, \forall w^s\in \mathcal{I}_{0}, \\
		&V_w v\leq \bm{1}, \	 \bm{\theta}\in [0,1]^{n_v}, \ \rho\in[0,1], ~\bm{\theta} \leq \rho\bm{1}.
    \end{aligned} \nonumber
\end{align}
As it is generally intractable to compute volumes of convex polytopes exactly, this problem aims to minimise an upper bound on ${\rm vol}(\mathcal{W}(v,\bm{\theta},\rho))$ provided in Lemma \ref{lemma:properties of parameterisation}. Although \eqref{quanset for homothetic mpc} is nonlinear due to the product of $\rho$ and $v$, the following proposition shows that it can be reformulated as an LP. 

\begin{proposition}\label{prop:reformulate scenario program}
Under Assumption \ref{Ass},	the optimal solution to problem \eqref{quanset for homothetic mpc} can be obtained by solving the following LP
\begin{align}
&\min_{y,\bm{\theta},\rho}  \bm{1}^T  \bm{\theta} + \rho   \label{quansetLP for homothetic mpc}\\
    &\begin{aligned} \text{ s.t. }~
        &V_w w^s \leq \bm{\theta} + V_w y, \forall w^s\in \mathcal{I}_{0}, \\
        &V_w y\leq (1-\rho)\bm{1}, \	 \bm{\theta}\in [0,1]^{n_v}, \ \rho\in[0,1], ~\bm{\theta} \leq \rho\bm{1}.
    \end{aligned}\nonumber
\end{align}
\end{proposition}
\begin{proof}
Any feasible point $(v,\bm{\theta},\rho)$ of problem \eqref{quanset for homothetic mpc} maps to a feasible point $(y,\bm{\theta},\rho)$ for \eqref{quansetLP for homothetic mpc} with the same objective value under the transformation $y\!=\!(1\!-\!\rho)v$. Likewise, the transformation $v\!=\!y/(1\!-\!\rho)$ maps any feasible point $(y,\bm{\theta},\rho)$ of problem \eqref{quansetLP for homothetic mpc} such that $\rho \in [0,1)$ to a feasible point $(v,\bm{\theta},\rho)$ of \eqref{quanset for homothetic mpc} with the same objective value, while a feasible point $(y,\bm{\theta},1)$ of \eqref{quansetLP for homothetic mpc} necessarily has $y=\bm{0}$ (since $\{w\in\mathbb{R}^{n_x}|V_w w \leq \bm{0}\}=\{\bm{0}\}$ following from the boundedness of $\mathbb{W}$ in Assumption \ref{Ass}) and therefore corresponds to a feasible point $(\bm{0},\bm{\theta},1)$ of \eqref{quanset for homothetic mpc} with the same objective value.
\end{proof}

We denote by $(y^*_0,\bm{\theta}^*_0,\rho^*_0)$ the optimiser of \eqref{quansetLP for homothetic mpc}. Then the optimal solution to \eqref{quanset for homothetic mpc} is given by $(v^*_0,\bm{\theta}^*_0,\rho^*_0)$, where $v_0^*=y^*_0/(1-\rho^*_0)$ if $\rho^*_0\neq 1$ and otherwise $v_0^*=\bm{0}$. We can now define the learned disturbance set at the initial time step as
\[ \hat{\mathbb{W}}_{0}:=\mathcal{W}(v^*_0,\bm{\theta}_0^*,\rho^*_0). \]
Next we provide a statistical gap between $\hat{\mathbb{W}}_{0}$ and $\mathbb{W}_{\rm true}$. 

\begin{theorem}[\!\!{\cite[Theorem 4]{5531078}}]\label{Theo:homotheticrisk}
	Suppose Assumptions \ref{Ass} and \ref{assumption: iid samples} hold. Given $\epsilon\in (0,1)$,  $\delta\in (0,1)$, and  Euler's constant $\rm{e}$, if
	$ 
		|\mathcal{I}_0|\geq \frac{1}{\epsilon} \frac{\rm{e}}{\rm{e}-1} \left( n_x+n_v+\ln\frac{1}{\delta}  \right)$,
	then, with probability no less than $1-\delta$, we have that
	$  
		{\rm Pr}[w\in \mathbb{W}_{\rm true}: w\notin \hat{\mathbb{W}}_{0}]\leq  \epsilon$.
\end{theorem}

\subsection{Online Update}
After initialisation of the learned disturbance set, we update it at each time step by exploiting disturbance realisations observed online to improve our approximation for $\mathbb{W}_{\rm true}$. The learned disturbance sets $\hat{\mathbb{W}}_k$ are obtained by solving the following problem at time steps $k=1,2,\ldots$
\begin{subequations}\label{eq:online update of learned disturbance set}
\begin{align}
    &\hspace{-1.0cm}(v^*_k,\bm{\theta}^*_k,\rho^*_k):=\argmin_{v,\bm{\theta},\rho} \bm{1}^T \bm{\theta} + \rho\\
        ~~\text{ s.t. } &\bm{\theta}^*_{k-1} + (1\!-\!\rho^*_{k-1})V_w v^*_{k-1} \leq \bm{\theta} + (1\!-\!\rho)V_w v, \label{eq: set inclusion among learned disturbance sets}\\
     &V_w w_{k-1} \leq \bm{\theta} + (1\!-\!\rho)V_w v,   \label{eq:new sample in the set}\\
     &V_w v\leq \bm{1}, \ \bm{\theta}\in [0,1]^{n_v}, \ \rho\in[0,1], ~\bm{\theta} \leq \rho\bm{1},
    \end{align}
\end{subequations}
and we define $\hat{\mathbb{W}}_{k}:=\mathcal{W}(v^*_k,\bm{\theta}_k^*,\rho^*_k)$. Constraint \eqref{eq:new sample in the set} ensures that the most up-to-date realisation of disturbance is contained in $\hat{\mathbb{W}}_{k}$. The constraint \eqref{eq: set inclusion among learned disturbance sets} ensures that $\hat{\mathbb{W}}_{k-1}\subseteq \hat{\mathbb{W}}_{k}$ and, together with \eqref{eq:new sample in the set}, avoids imposing that $w^s\in \mathcal{W}(v,\bm{\theta},\rho)$ for all $w^s\in \mathcal{I}_k$, which would lead to increasing computational complexity over time. Problem \eqref{eq:online update of learned disturbance set} can be reformulated equivalently as an LP using the same reasoning as in Proposition \ref{prop:reformulate scenario program}. Given the nestedness property of learned disturbance sets, $\hat{\mathbb{W}}_k$ for $k\in\mathbb{N}_{\geq 1}$ enjoys the same level of approximation accuracy as $\hat{\mathbb{W}}_0$, as shown in Theorem \ref{thm: probabilistic guarantee for online updated learned disturbance sets}.
\begin{theorem}\label{thm: probabilistic guarantee for online updated learned disturbance sets}
	Suppose Assumptions \ref{Ass} and \ref{assumption: iid samples} hold. Given $\epsilon\in (0,1)$,  $\delta\in (0,1)$, and  Euler's constant $\rm{e}$, if
	$ 
		|\mathcal{I}_0|\geq \frac{1}{\epsilon} \frac{\rm{e}}{\rm{e}-1} \left( n_x+n_v + \ln\frac{1}{\delta}  \right)$,
	then, with probability no less than $1-\delta$, we have that
	$ 
		{\rm Pr}[w\in \mathbb{W}_{\rm true}: w\notin \hat{\mathbb{W}}_{k}]\leq  \epsilon, ~\forall k \in \mathbb{N}
	$.
\end{theorem}
This result is a direct consequence of Theorem \ref{Theo:homotheticrisk} and the fact that $\hat{\mathbb{W}}_{k}\subseteq \hat{\mathbb{W}}_{k+1}$ holds for all $k\in\mathbb{N}$.

\section{Learning-based Homothetic Tube MPC} \label{sec: Learning-based MPC}
In this section, we exploit sets $\hat{\mathbb{W}}_{k}$ and design a homothetic tube MPC algorithm, where cross-sections of the tube are uniform scalings of the set $\mathbb{S}$ satisfying Assumption \ref{ass:robust invariance}. Recalling preliminaries in Section \ref{sec: preliminary}, we formulate the MPC optimisation based on $\hat{\mathbb{W}}_{k}$ at time step $k$ as
\begin{eqnarray}\label{Eq:QuanHomTubeMPCOPT}
 \begin{aligned}
     &\min\limits_{s_{0|k},\bm{c}_k,\bm{\alpha}_k\in\mathbb{R}_{+}^N} J(s_{0|k}, \bm{c}_k,\bm{\alpha}_k)    \\
		&{~\rm s.t.}
		\begin{cases}
			x_k-s_{0|k}\in \alpha_{0|k}\mathbb{S}, \\
			\alpha_{i|k} e_{\rm max} \!+\hat{w}_{{\rm max},k} \!\leq\!   \alpha_{i+1|k}\bm{1}, \forall i\in\mathbb{N}_{[0,N-1]},  \\
			\alpha_{i|k}=1,  \forall i\in\mathbb{N}_{\geq N}, \\
                \bar{F}\Psi^i \!\begin{bmatrix}
				s_{0|k} \\
				\bm{c}_k
			\end{bmatrix}\!\!\leq\! \bm{1}-\alpha_{i|k}h, 	 \forall i\in\mathbb{N}_{[0,\nu]}.
		\end{cases}
 \end{aligned}
\end{eqnarray}
This problem is denoted by ${\rm OPT}(\hat{\mathbb{W}}_k, \hat{w}_{{\rm max},k},\nu)$, where 
\begin{equation}\label{eq:update w_max}
    \hat{w}_{{\rm max},k}:=\max_{w\in \hat{\mathbb{W}}_k}  V_s w,
\end{equation}
and $\nu$ is no less than $N-1$ and is the smallest integer such that $\bar{F}\Psi^{\nu+1}z \leq \bm{1}-h$ holds for all $z$ and $\bm{\alpha}:=(\alpha_0,\ldots,\alpha_{N-1})$ in the set 
\begin{align*}
    &\Omega(\hat{w}_{{\rm max},0},\nu):= \\
    &\negmedspace{}\bigl\{(z,\bm{\alpha})\in \mathbb{R}^{n_x + N n_u}\!\times \mathbb{R}_+^N |\,\bar{F}\Psi^{i}z \leq \bm{1}-\alpha_i h, \forall i \in \mathbb{N}_{[0,\nu]}, \\
    &\alpha_i e_{\rm max} \!+\! \hat{w}_{{\rm max},0} \!\leq\! \alpha_{i+1} \bm{1}, \forall i \!\in\! \mathbb{N}_{[0,N-1]}, \ \alpha_i\!=\!1, \forall i\!\in\!\mathbb{N}_{\geq N}    \!\bigr\}.
\end{align*}

In problem \eqref{Eq:QuanHomTubeMPCOPT}, it remains a feasible choice to limit the number of variables by setting $\alpha_{i|k}=1$ for all $i\in \mathbb{N}_{\geq N}$. Specifically, it holds true that $e_{\rm max} +\hat{w}_{{\rm max},k} \leq  \bm{1}$ for all $k\in\mathbb{N}$. This follows from that $\mathbb{S}$ satisfies Assumption \ref{ass:robust invariance} and $\hat{\mathbb{W}}_k \subseteq \mathbb{W}$ for all $k\in\mathbb{N}$ as shown in Lemma \ref{lemma:properties of parameterisation}.

\subsection{Effect of $\nu$ on Constraint Handling}
We next elaborate on the design of integer $\nu$ in \eqref{Eq:QuanHomTubeMPCOPT} and a desirable result it yields. Analogously to the definition of $\nu$, we define $\nu_k$ as the integer that is no less than $N-1$ and is the smallest such that $\max_{(z,\bm{\alpha})\in\Omega(\hat{w}_{{\rm max},k},\nu_k)}\bar{F}\Psi^{\nu_k+1}z \leq \bm{1}-h$ holds, and $\nu_0=\nu$. By Lemma \ref{lemma:condition on tau}, such an integer $\nu_k$ ensures that $\Omega(\hat{w}_{{\rm max},k},\nu_k)=\Omega(\hat{w}_{{\rm max},k},\infty)$ and avoids imposing in ${\rm OPT}(\hat{\mathbb{W}}_k, \hat{w}_{{\rm max},k},\nu_k)$ that $\bar{F}\Psi^i \!\begin{bmatrix}
				s_{0|k} \\
				\bm{c}_k
			\end{bmatrix}\!\!\leq\! \bm{1}-\alpha_{i|k}h$ for all $i\in\mathbb{N}$, which is a set of infinitely many constraints. However, it is expected that $\nu_k$ depends on $\hat{w}_{{\rm max},k}$, which could be changing over time, and this requires that $\nu_k$ is recomputed online at each time step $k$. The fixed integer $\nu$ addresses this issue and therefore reduces the online computational burden, as shown in the following result. 
   \begin{lemma}
   Under Assumption \ref{Ass}, it holds that $\nu_{k+1}\leq \nu_k$ for all $k\in\mathbb{N}$ and therefore problem ${\rm OPT}(\hat{\mathbb{W}}_k, \hat{w}_{{\rm max},k},\nu)$ is equivalent to ${\rm OPT}(\hat{\mathbb{W}}_k, \hat{w}_{{\rm max},k},\nu_k)$ for all $k\in\mathbb{N}$.
   \end{lemma}
   \begin{proof}
        By Assumption \ref{Ass}, $\mathbb{W}$ is compact. Then, $\hat{\mathbb{W}}_k$ is compact and $\hat{w}_{{\rm max},k}$ is well defined for all $k\in\mathbb{N}$. The nestedness property of $\hat{\mathbb{W}}_k$ implies that $\hat{w}_{{\rm max},k}\leq \hat{w}_{{\rm max},k+1}$ for all $k\in\mathbb{N}$, which yields that
        \begin{equation}\label{eq: set containment of Omega}
            \Omega(\hat{w}_{{\rm max},k+1},n)\subseteq \Omega(\hat{w}_{{\rm max},k},n),~\forall k\in\mathbb{N},
        \end{equation}
        where $n$ is an integer that is no less than $N-1$. By the definition of $\nu_k$, it holds that $\max_{(z,\bm{\alpha})\in\Omega(\hat{w}_{{\rm max},k},\nu_k)}\bar{F}\Psi^{\nu_k+1}z \leq \bm{1}-h$, which, together with \eqref{eq: set containment of Omega}, implies that 
        \begin{equation}\label{eq: ineq relate nu_k to nu_k+1}
            \max_{(z,\bm{\alpha})\in\Omega(\hat{w}_{{\rm max},k+1},\nu_k)}\bar{F}\Psi^{\nu_k+1}z \leq \bm{1}-h. 
        \end{equation}
        As $\nu_{k+1}$ is the smallest integer no less than $N-1$ such that \eqref{eq: ineq relate nu_k to nu_k+1} holds, we have that $\nu_{k+1} \leq \nu_k \leq \nu $ for all $k\in\mathbb{N}$. Then, the desired result follows from Lemma \ref{lemma:condition on tau}.  
   \end{proof}
   
We note that $\nu$ can be obtained of\mbox{}f\mbox{}line given $h$ and $\hat{w}_{{\rm max},0}$. Since the sequence $\{\nu_k\}_{k=0}^\infty$ is monotonically non-increasing, one can replace $\nu$ in problem \eqref{Eq:QuanHomTubeMPCOPT} at a later stage with a feasible and possibly smaller integer to further reduce the number of constraints, if computational resources permit. 

\subsection{Probabilistic Recursive Feasibility}
Given the design of integer $\nu$, we leverage Lemma \ref{lemma:conventional recursive feasibility} and analyse recursive feasibility of problem \eqref{Eq:QuanHomTubeMPCOPT} from a probabilistic point of view.
\begin{proposition}
    Suppose conditions in Theorem \ref{thm: probabilistic guarantee for online updated learned disturbance sets} are satisfied. If problem \eqref{Eq:QuanHomTubeMPCOPT} is feasible at time step $k$, then, with confidence no less than $1-\delta$, it is feasible at time step $k+1$ with probability at least $1-\epsilon$.
\end{proposition}

This result follows directly from Lemma \ref{lemma:conventional recursive feasibility} and Theorem \ref{thm: probabilistic guarantee for online updated learned disturbance sets}. Moreover, if it holds that $\mathbb{W}_{\rm true}\subseteq\hat{\mathbb{W}}_k$ at some time step $k$ and problem \eqref{Eq:QuanHomTubeMPCOPT} is feasible at this time step, then it remains feasible at all time steps after that.

\subsection{Computation}
 The computation at each time step involves solving the LP \eqref{eq:online update of learned disturbance set} for learning the uncertainty set, the LP~\eqref{eq:update w_max} for worst-case estimation, and the QP~\eqref{Eq:QuanHomTubeMPCOPT} for MPC. These problems are solved by polynomial-time algorithms. 

\section{Numerical Examples}\label{sec: examples}

\begin{figure}[t]
\centering
\subfloat[\label{Fig: disturbance set}]{\includegraphics[width = 0.5\columnwidth]{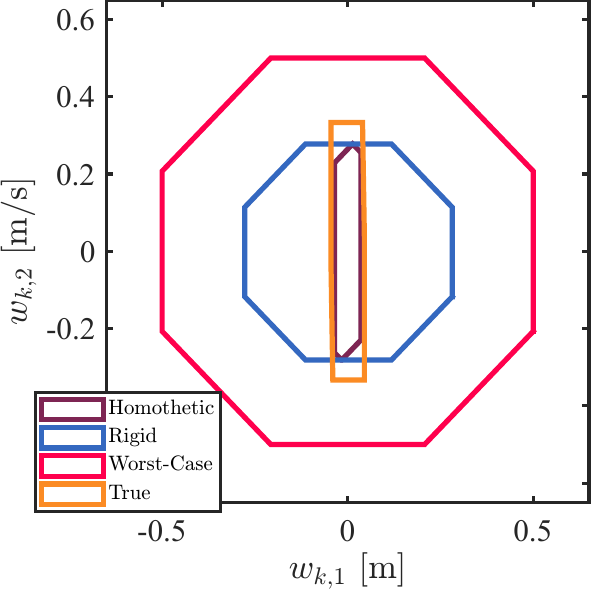}} 
\centering
\subfloat[\label{Fig: feasible region}]{\includegraphics[width = 0.5\columnwidth]{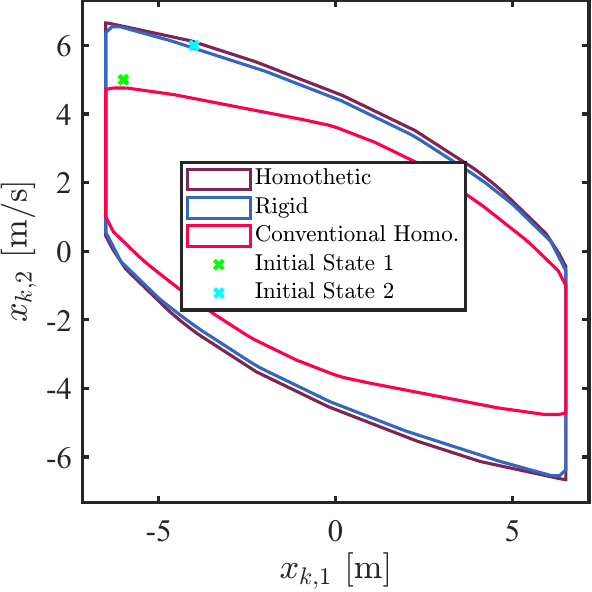}}
\caption{The disturbance set and feasible region. (a) Learned disturbance set by homothetic approach and rigid approach of~\cite{gao2023learning} with $|\mathcal{I}_0|=30$, the worst-case disturbance set, and the true disturbance set; (b) The initial feasible region of the learning-based homothetic MPC, the learning-based rigid MPC, and the conventional homothetic MPC with $|\mathcal{I}_0|=30$.}
\label{Fig: disturbance set and feasible region}
\end{figure}

\begin{figure}[t]
\centering
\subfloat[\label{Fig: path feasible}]{\includegraphics[width = 0.5\columnwidth]{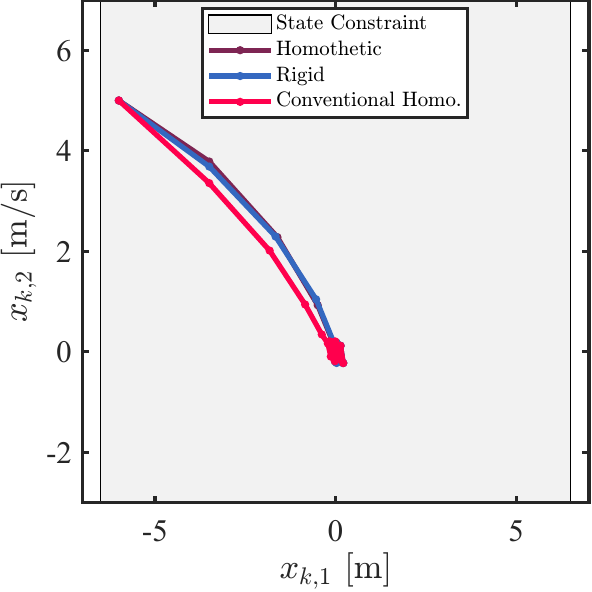}} 
\centering
\subfloat[\label{Fig: path infeasible}]{\includegraphics[width = 0.5\columnwidth]{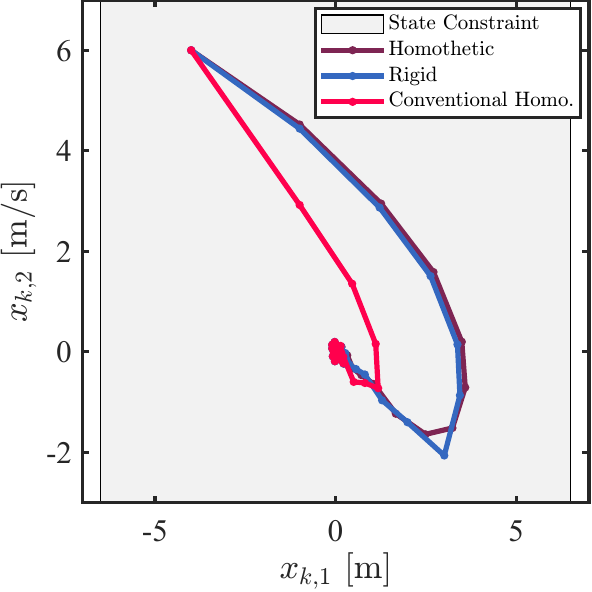}}
\caption{Trajectories and control inputs of the relative model. (a) State trajectories for three methods with the initial state as state 1 in Fig.~\ref{Fig: disturbance set and feasible region}(b); (b) State trajectory for three methods with the initial state as state 2 in Fig.~\ref{Fig: disturbance set and feasible region}(b).}
\label{Fig:path}
\end{figure}

\begin{table}[t]
\centering
\caption{Online Computation Performance}  
\label{tab: online computation performance} 
\begin{tabular}{cccc}
\toprule
\textbf{Method} & \textbf{Mean} & \textbf{STD} \\
\midrule
Learning-based Homothetic Tube MPC & $0.24 \ {\rm s}$ & $0.16 \ {\rm s}$  \\ \specialrule{0em}{1pt}{1pt}
Learning-based Rigid Tube MPC & $ 0.41 \ {\rm s}$ & $0.25 \ {\rm s}$  \\ \specialrule{0em}{1pt}{1pt} 
Conventional Homothetic Tube MPC  & $0.03 \ {\rm s}$ & $0.02 \ {\rm s}$  \\ \specialrule{0em}{1pt}{1pt}
\bottomrule
\end{tabular}
\begin{tablenotes}\tiny
	\item ${\star}$ STD means standard deviation. 
\end{tablenotes} 
\end{table}

\begin{table}[!t]
\centering
\caption{Empirical Robustness}  
\label{tab: empirical robustness} 
\begin{tabular}{cccc}
\toprule
$|\mathcal{I}_0|$ & Initial Relative State & Feasibility Rate & $\epsilon$\\
\midrule
100 & $[-6.5 \ 6.6831]^{\top}$ & $100 \%$ & $0.1739$ \\ \specialrule{0em}{1pt}{1pt}
1000 & $[-6.5 \ 6.5757]^{\top}$ & $100 \%$ & $0.0174$ \\ \specialrule{0em}{1pt}{1pt}
10000 & $[-6.5 \ 6.4529]^{\top}$ & $100 \%$ & $0.0017$ \\ \specialrule{0em}{1pt}{1pt}
\bottomrule
\end{tabular}
\begin{tablenotes}\tiny
\item $\star$ Parameter $\epsilon$ is computed according to Theorem~\ref{thm: probabilistic guarantee for online updated learned disturbance sets} with $\delta$ fixed as $0.05$.
\end{tablenotes} 
\end{table}

\begin{figure}
\centering
\subfloat[\label{Fig: path feasible 5V}]{\includegraphics[width =\columnwidth]{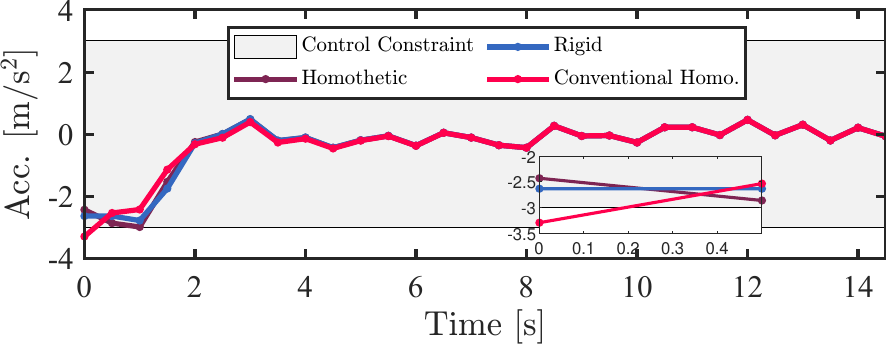}} \\
\centering
\subfloat[\label{Fig: path infeasible 5V}]{\includegraphics[width =1 \columnwidth]{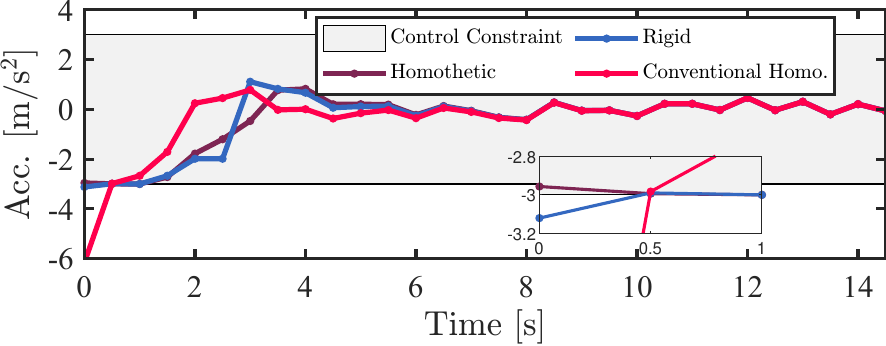}}
\caption{(a) Control inputs of the follower agent by three methods corresponding to the state trajectories in Fig.~\ref{Fig:path}(a); (b) Control inputs of the follower agent by three methods corresponding to the state trajectories in Fig.~\ref{Fig:path}(b).}
\label{Fig:control}
\end{figure}

In this section, we evaluate the efficacy of the proposed algorithm in a platooning-control problem of two agents, where a follower in the platooning needs to follow its leader and maintains a desired relative speed and relative distance. Their dynamics can be described by linear time-invariant models with superscripts $f$ and $l$ respectively, i.e.%
\begin{subequations}\label{eq: leader and follower models}
\begin{align}
x_{k+1}^f &= Ax_{k}^f + Bu_{k}^f + \xi_{k}^f,\\
x_{k+1}^l &= Ax_{k}^l + Bu_{k}^l+ \xi_{k}^l,
\end{align}
\end{subequations}
where $x_{k}^f := [p_{k}^f \ v_{k}^f]^{T}$ contains the longitudinal position $p_{k}^f$ and velocity $v_k^f$ of the follower, $x_{k}^l := [p_{k}^l \ v_{k}^l]^{ T}$ contains the longitudinal position $p_k^l$ and velocity $v_k^l$ of the leader, and $u_{k}^f:=a_{k}^f$ and $u_{k}^l:=a_{k}^l$ are their longitudinal accelerations. In addition, $A := [\begin{smallmatrix} 1 & t_s\\ 0 & 1\end{smallmatrix}]$, $B := [\begin{smallmatrix} 0 \\ t_s\end{smallmatrix}]$, where $t_s$ is a sampling interval, and $\xi_{k}^f \in \mathbb{R}^2$ and $\xi_{k}^l \in \mathbb{R}^2$ represent modelling uncertainty and external disturbance. We denote the true disturbance sets for $\xi_{k}^l$ and $\xi_{k}^f$ as $\Xi^l_{\rm true}$ and $\Xi^f_{\rm true}$, respectively. Also, $u_{k}^l$ is constrained to lie in the set $\mathbb{U}^l_{\rm true}\subset \mathbb{R}$. Note that $\Xi^l_{\rm true}$, $\Xi^f_{\rm true}$,  and $\mathbb{U}^l_{\rm true}$  are unknown to the follower. Defining $x^{\rm des} := [-L \ 0]^{T}$ with a desired safety distance $L$, $x_{k} := x_{k}^f-x_{k}^l-x^{\rm des}$, and $u_k := a_k^f$,  we derive the relative dynamics as 
\begin{equation}\label{eq: relative model}
x_{k+1} = Ax_{k} + Bu_{k} + w_{k},
\end{equation}
where $w_{k} :=  \xi_{k}^f-Bu_{k}^l-\xi_{k}^l \in \mathbb{W}_{\rm true}:= \Xi^f_{\rm true} \oplus  (-\Xi^l_{\rm true}) \oplus (-B\mathbb{U}^l_{\rm true})$. In this platooning-control problem, the objective is to regulate the state in \eqref{eq: relative model} such that it converges to the origin while minimising control efforts. 

Challenges arise in the platooning-control problem, as the disturbance set $\mathbb{W}_{\rm true}$ is unknown to the controlled system. We compare our method with alternative approaches via a set of simulations. The optimisation problems involved are solved by \texttt{Ipopt}~\cite{wachter2006implementation} in \texttt{CasADi}~\cite{andersson2019casadi}, and set calculations by \texttt{Multi-Parametric Toolbox 3.0}~\cite{herceg2013multi} and \texttt{Yalmip} \cite{lofberg2004yalmip}. Also, all the parameters of these implementations can be found in our published code\footnote{\textcolor{blue}{\tiny{\url{https://github.com/JianZhou1212/learning-based-homothetic-tube-mpc}}}}. 

\emph{Learned Disturbance Sets and Corresponding Feasible Regions:} We first evaluate approximations to $\mathbb{W}_{\rm true}$ resulting from the methods in Section~\ref{sec: learning uncertainty set} and \cite{gao2023learning} and plot them in Fig.~\ref{Fig: disturbance set and feasible region} (a). The latter method provides an approximation as uniform scaling of the worst-case disturbance set with a translation. We then compare initial feasible sets of the learning-based homothetic and rigid tube MPC problems (where the aforementioned approximation sets are used to construct homothetic and rigid tubes, respectively) and their counterpart of the conventional homothetic tube MPC problem \eqref{Eq:HomTubeMPCOPT} (where the worst-case disturbance set is used to build homothetic tubes). These sets are shown in Fig.~\ref{Fig: disturbance set and feasible region} (b).

As shown in Fig.~\ref{Fig: disturbance set and feasible region} (a), the learned disturbance sets obtained by the method in Section~\ref{sec: learning uncertainty set} and the one in \cite{gao2023learning} are considerably less conservative than the worst-case disturbance set. This results in larger initial feasible regions for the learning-based homothetic and rigid tube MPC problems, as seen in Fig.~\ref{Fig: disturbance set and feasible region} (b). Also, since the set parameterisation of the latter method has fewer degrees of freedom, its corresponding learned disturbance set over-approximates the true set in some directions, in comparison with the method in Section~\ref{sec: learning uncertainty set}. This difference is also reflected in Fig.~\ref{Fig: disturbance set and feasible region} (b).

\emph{Online Performance:} To verify enhanced feasibility of the proposed method, we select two initial states shown in Fig.~\ref{Fig: disturbance set and feasible region} (b) for the relative model~\eqref{eq: relative model} and control the system with three algorithms. One initial state is near the boundary of the feasible region of the learning-based homothetic tube MPC while outside the feasible region of the other two approaches. Another initial state is within the feasible region of the learning-based rigid tube MPC, while outside that of the conventional homothetic tube MPC. 
State and input trajectories generated by these three algorithms are shown in Fig.~\ref{Fig:path} and Fig.~\ref{Fig:control}, respectively.
It is seen that the proposed approach satisfies control constraints while starting from both initial states. On the other hand, the learning-based rigid tube MPC violates control constraints while starting from initial state 2, and the conventional homothetic tube MPC is infeasible in both cases.

\emph{Online Computation:} The computation time of three control algorithms is shown in Table~\ref{tab: online computation performance}. It is seen that our proposed method is more efficient than the learning-based rigid tube MPC \cite{gao2023learning}. As in this work it is not necessary to recompute at each time step a finite number of constraints to be enforced, i.e. $\nu$ in \eqref{Eq:QuanHomTubeMPCOPT} is fixed, even though learned disturbance sets are possibly time varying. However, comparing with the conventional homothetic tube MPC, our proposed method has greater computational load, which is the cost of improved feasibility. During its implementation, the majority of computation time is spent in calculating $\hat{w}_{{\rm max},k}$ in~\eqref{eq:update w_max} at each time step. This is expected to be improved by fully utilising parallel computing capabilities.

\emph{Empirical Robustness:} We run $100$ simulations starting from the same initial condition for each value of $|\mathcal{I}_0|$ listed in Table~\ref{tab: empirical robustness}. The selected initial conditions of model~\eqref{eq: relative model} are at the boundaries of the feasible regions. Feasibility rates of these simulations are shown in Table~\ref{tab: empirical robustness}. Our proposed method achieves $100\%$ feasibility in all sets of simulations, and feasibility at all time steps of a simulation implies closed-loop constraint satisfaction.

\section{Conclusion} \label{sec: conclusion}
In this paper, we propose a learning-based homothetic tube MPC strategy for discrete-time linear systems
subject to bounded additive disturbance and mixed constraints on the state and input. Assuming that the true disturbance
set is unknown, we provide a data-driven algorithm to learn online the true disturbance set at each time step. The idea is to find an optimal scaling of a conservative set 
along different directions, which can be efficiently solved by an LP. 
We characterise statistical gaps between the true and learned disturbance sets, 
and on the basis of this we investigate probabilistic recursive feasibility of the learning-based MPC. Numerical simulations demonstrate
the efficacy and empirical robustness of our proposed algorithm and compare with the state-of-the-art MPC algorithms.

  \bibliographystyle{ieeetr}
  \bibliography{reference}
\end{document}